\documentclass[12pt,a4paper]{article}
\usepackage{amsmath}
\usepackage{amsfonts}
\usepackage{amssymb}
\usepackage{makeidx}
\usepackage[pdftex]{graphicx}
\usepackage[applemac]{inputenc}
\usepackage{amssymb}
 \usepackage{amsthm}
  \usepackage{amsmath}
 \usepackage{textcomp}
\usepackage{pgfplots}
 \usepackage{tikz}
 \usepackage{lineno}

\usepackage{setspace}

\begin{document}

\doublespacing

\title{Gouy-Stodola Theorem as a variational principle for open systems.}

\author{Umberto Lucia\\ Dipartimento Energia, Politecnico di Torino, \\ Corso Duca degli Abruzzi 24, 10129 Torino, Italy\\ umberto.lucia@polito.it}
\date{}
\maketitle

\begin{abstract}
The recent researches in non equilibrium and far from equilibrium systems have been proved to be useful for their applications in different disciplines and many subjects. A general principle to approach all these phenomena with a unique method of analysis is required in science and engineering: a variational principle would have this fundamental role. Here, the Gouy-Stodola theorem is proposed to be this general variational principle, both proving that it satisfies the above requirements and relating it to a statistical results on entropy production.
\end{abstract}
\maketitle

\section{Introduction}
\label{intro}
The variational methods are very important in mathematics, physics, science and engineering because they allow to describe the natural systems by physical quantities independently from the frame of reference used \cite{ozisik1980}. Moreover, Lagrangian formulation can be used in a variety of physical phenomena and a structural analogy between different physical phenomena has been pointed out \cite{truesdell}. The most important result of the variational principle consists in obtaining both local and global theories \cite{arnold}: global theory allows us to obtain information directly about the mean values of the physical quantities, while the local one about their distribution \cite{lucia1995}. 

Last, the notions of entropy and its production and generation are the fundamentals of modern thermodynamics and a lot of variational approaches has been proposed in thermodynamics \cite{martyushev}. Today, the researches in non equilibrium and far from equilibrium systems have been proved to be useful for their application in mathematical and theoretical biology, biotechnologies, nanotechnologies, ecology, climate changes, energy optimization, thermo-economy, phase separation, aging process, system theory and control, pattern formation, cancer pharmacology, DNA medicine, metabolic engineering, chaotic and dynamical systems, all summarized in non linear, dissipative and open systems. A general principle to approach all these phenomena with a unique method of analysis is required in science and engineering: a variational principle would have this fundamental role. But, a useful variational principle out of the use of work has been proved not to be obtained \cite{lebowitz} and, for the open systems, the entropy production for the system and the reservoir has been proved to be obtained if and only if only one heat bath \cite{lebowitz} exists and its temperature is constant \cite{lebowitz}. The different variational principle have never been related with these fundamental requirement for a general principle in thermodynamics of open system, and often they are related only to closed or isolated systems \cite{martyushev}.
 
Last, in engineering thermodynamics the open systems are analysed using entropy generation, but its statistical definition does not exist and this lack does not allow to extend this powerful method to all the above discipline, where statistical approach is required. On the other hand, the statistical models are not so general to be used in any case with a unique approach; consequently, no of them represents a general principle of investigation.

In this paper, the Gouy-Stodola theorem \cite{gouya},\cite{gouyb},\cite{duhem},\cite{gouyc} is proposed to be this general variational principle, both proving that it satisfies the above requirements and relating it to a statistical results on entropy production.

\section{The open systems}
\label{sec:1}
%and \cite{Ref1}
In this Section the thermodynamic system is defined. To do so, the definition of `system with perfect accessibility', which allows to define both the thermodynamic and the dynamical systems, must be considered.

An open $N$ particles system is considered. Every $i-$th element of this system is located by a position vector $\mathbf{x}_{i}\in R^{3}$, it has a velocity $\dot{\mathbf{x}}_{i}\in R^{3}$, a mass $m_{i}\in R$ and a momentum $\mathbf{p}=m_{i}\mathbf{\dot{x}}_{i}$, with $i\in [1,N]$ and $\mathbf{p}\in R^{3}$ \cite{lucia2008}. The masses $m_{i}$ must satisfy the condition:
	\begin{equation}
		\sum_{i=1}^{N}m_{i}=m
	\end{equation} 
where $m$ is the total mass which must be a conserved quantity, so it follows:
	\begin{equation}
		\dot{\rho}+\rho \nabla\cdot\dot{\mathbf{x}}_{B}=0
	\end{equation} 
where $\rho =dm/dV$ is the total mass density, with $V$ total volume of the system and $\dot{\mathbf{x}}_{B}\in R^{3}$, defined as $\dot{\mathbf{x}}_{B}=\sum_{i=1}^{N}\mathbf{p}_{i}/m$, velocity of the centre of mass. The  mass density must satisfy the following conservation law \cite{lucia1995}:
	\begin{equation}
		\dot{\rho}_{i}+\rho_{i} \nabla\cdot\dot{\mathbf{x}}_{i}=\rho \Xi
	\end{equation} 
where $\rho_{i}$ is the density of the $i-$th elementary volume $V_{i}$, with $\sum_{i=1}^{N}V_{i}=V$, and $\Xi$ is the source, generated by matter transfer, chemical reactions and thermodynamic transformations. This open system can be mathematical defined as follows.

\newtheorem{definition}{Definition}
\begin{definition} \label{ph} \cite{huang1987}
 - A dynamical state of $N$ particles can be specified by the $3N$ canonical coordinates $\left\{\mathbf{q}_{i}\in R^{3},i\in \left[1,N\right]\right\}$ and their conjugate momenta $\left\{\mathbf{p}_{i}\in R^{3},i\in \left[1,N\right]\right\}$. The $6N-$dimensional space spanned by $\left\{\left(\mathbf{p}_{i},\mathbf{q}_{i}\right),i\in \left[1,N\right]\right\}$ is called the phase space $\Omega$. A point $\mathbf{\sigma}_{i} =\left(\mathbf{p}_{i},\mathbf{q}_{i}\right)_{i\in \left[1,N\right]}$ in the phase space $\Omega :=\left\{\mathbf{\sigma}_{i}\in R^{6N}:\mathbf{\sigma}_{i}=\left(\mathbf{p}_{i},\mathbf{q}_{i}\right), i\in \left[1,N\right]\right\}$ represents a state of the entire $N-$particle system.
\end{definition}

\begin{definition} \cite{lucia2008}
 - A system with perfect accessibility $\Omega_{PA}$ is a pair $\left(\Omega, \Pi\right)$, with $\Pi$ a set whose elements  $\pi$ are called process generators, together with two functions:
\begin{equation}
\pi \mapsto \mathcal{S}
\end{equation}
\begin{equation}
\left(\pi^{'},\pi{''}\right) \mapsto \pi^{''}\pi{'}
\end{equation}
where $\mathcal{S}$ is the state transformation induced by $\pi$, whose domain $\mathcal{D}\left(\pi\right)$ and range $\mathcal{R}\left(\pi\right)$ are non-empty subset of $\Omega$.
This assignment of transformation to process generators is required to satisfy the following conditions of accessibility:
\begin{enumerate}
	\item $\Pi\sigma:=\left\{\mathcal{S}\sigma :\pi\in\Pi,\sigma\in \mathcal{D}\left(\pi\right)\right\}=\Omega$ , $\forall \sigma\in\Omega$\emph{:} the set $\Pi\sigma$ is called the set of the states accessible from $\sigma$ and, consequently, it is the entire \emph{state space}, the phase spase $\Omega$;
	\item if $\mathcal{D}\left(\pi ''\right)\cap \mathcal{R}\left(\pi '\right)\neq 0\Rightarrow \mathcal{D}\left(\pi ''\pi '\right)=\mathcal{S}_{\pi '}^{-1}\left(\mathcal{D}\left(\pi ''\right)\right)$ and $\mathcal{S}_{\pi ''\pi '}\sigma =\mathcal{S}_{\pi ''}\mathcal{S}_{\pi '}\sigma$ $\forall\sigma\in \mathcal{D}\left(\pi ''\pi '\right)$
\end{enumerate}
\end{definition}

\begin{definition} \cite{lucia2008} 
- A process in $\Omega_{PA}$ is a pair $\left(\pi ,\sigma\right)$, with $\sigma$ a state and $\pi$ a process generator such that $\sigma$ is in $\mathcal{D}\left(\pi\right)$. The set of all processes of $\Omega_{PA}$ is denoted by:
\begin{equation}
\Pi \diamond\Omega =\left\{\left(\pi ,\sigma\right): \pi \in \Pi ,\sigma\in \mathcal{D}\left(\pi\right)\right\}
\end{equation}
If $\left(\pi ,\sigma\right)$ is a process, then $\sigma$ is called the initial state for $\left(\pi ,\sigma\right)$ and $\mathcal{S}\sigma$ is called the final state for $\left(\pi ,\sigma\right)$.
\end{definition}

\begin{definition}
In an open system, there exists a characteristic time of any process, called lifetime $\tau$, which represents the time of evolution of the system between two stationary states.
\end{definition}

Any observation of the open irreversible system, in order to evaluate its physical quantities related to stationary states, must be done only at the initial time of the process and at its lifetime. During this time range the system moves through a set of non-equilibrium states by fluctuations on its thermodynamic paths.

\begin{definition} \cite{lucia1995} 
\label{thermdef} - A thermodynamic system is a system with perfect accessibility $\Omega_{PA}$ with two actions $W\left(\pi ,\sigma\right)\in R$ and $H\left(\pi ,\sigma\right)\in R$, called work done and heat gained by the system during the process $\left(\pi ,\sigma\right)$, respectively.
\end{definition}

The set of all these stationary states of a system $\Omega_{PA}$ is called non-equilibrium ensemble \cite{lucia2008}.

\begin{definition} \cite{lucia2008}
- A thermodynamic path $\gamma$ is an oriented piecewise continuously differentiable curve in $\Omega_{PA}$.
\end{definition}

\begin{definition} \cite{lucia2008} 
- A cycle $\mathcal{C}$ is a path whose endpoints coincide.
\end{definition}

\begin{definition} \cite{gallavotti2003}
- A smooth map $\mathcal{S}$ of a compact manifold $\mathcal{M}$ is a map with the property that around each point $\mathbf{\sigma}$ it can be established a system of coordinates based on smooth surfaces $W^{s}_{\mathbf{\sigma}}$ and $W^{u}_{\mathbf{\sigma}}$, with $s$=stable and $u$=unstable, of complementary positive dimension which is:
\begin{enumerate}
	\item covariant: $\partial\mathcal{S} W^{i}_{\mathbf{\sigma}}=W^{i}_{\mathcal{S}\mathbf{\sigma}}, i=u,s$. This means that the tangent planes $\partial\mathcal{S} W^{i}_{\mathbf{\sigma}}, i=u,s$ to the coordinates surface at $\mathbf{\sigma}$ are mapped over the corresponding planes at $\mathcal{S}\mathbf{\sigma}$;
	\item continuous: $\partial\mathcal{S} W^{i}_{\mathbf{\sigma}}$, with $i=u,s$, depends continuously on $\mathbf{\sigma}$;
	\item transitivity: there is a point in a subsistem of $\Omega_{PA}$ of zero Liouville probability, called attractor, with a dense orbit.
\end{enumerate}
\end{definition}

A great number of systems satisfies also the hyperbolic condition: the lenght of a tangent vector $\mathbf{v}$ is amplified by a factor $C\lambda ^{k}$ for $k>0$ under the action of $\mathcal{S}^{-k}$ if $\mathbf{\sigma}\in W^{s}_{k}$ with $C>0$ and $\lambda <1$. This means that if an observer moves with the point $\mathbf{\sigma}$ it sees the nearby points moving around it as if it was a hyperbolic fixed point. But, in a general approach this

The experimental observation allows to obtain and measure the macroscopic quantities which are mathematically the consequence of a statistics $\mu_{E}$ describing the asymptotic behaviour of almost all initial data in perfect accessibility phase space $\Omega_{PA}$ such that, except for a volume zero set of initial data $\mathbf{\sigma}$, it will be:
\begin{equation}
\lim_{T\longrightarrow\infty}\frac{1}{T}\sum^{T-1}_{j=1}\varphi\left(\mathcal{S}^{j}\mathbf{\sigma}\right)=\int_{\Omega}\mu_{E}\left(d\mathbf{\sigma}\right)\varphi\left(\mathbf{\sigma}\right)
\end{equation}
for all continuous functions $\varphi$ on $\Omega_{PA}$ and for every transformation $\mathbf{\sigma}\mapsto \mathcal{S}_{t}\left(\mathbf{\sigma}\right)$. For hyperbolic systems the distribution $\mu_{E}$ is the Sinai-Ruelle-Bowen distribution, SRB-distribution or SRB-statistics. In particular, here, the statistics is referred to a finite time $\tau$ process, as every real process is, so it is considered a SRB-statistics for a finite time system, which exists even if it is not so easy to be evaluated.

The notation $\mu_{E}\left(d\mathbf{\sigma}\right)$ expresses the possible fractal nature of the support of the distribution $\mu_{E}$, and implies that the probability of finding the dynamical system in the infinitesimal volume $d\mathbf{\sigma}$ around $\mathbf{\sigma}$ may not be proportional to $d\mathbf{\sigma}$  \cite{lucia2008}. Consequently, it may not be written as $\mu_{E}\left(\mathbf{\sigma}\right) d\mathbf{\sigma}$, but it needs to be conventionally expressed as $\mu_{E}\left(d\mathbf{\sigma}\right)$. The fractal nature of the phase space is an issue yet under debate \cite{garcia2006}, but there are a lot of evidence on it in the low dimensional systems \cite{hoover1998}. Here this possibility is also considered.

\begin{definition} \cite{billingsley1979} 
 - The triple $\left(\Omega_{PA},\mathcal{F},\mu_{E}\right)$ is a measure space, the Kolmogorov probability space, $\Gamma$.
\end{definition}

\begin{definition} \cite{lucia2008} 
 - A dynamical law $\tau_d$ is a group of meausure-preserving automorphisms $\mathcal{S}:\Omega_{PA}\rightarrow\Omega_{PA}$ of the probability space $\Gamma$.
\end{definition}

\begin{definition} \cite{lucia2008}
 - A dynamical system $\Gamma_{d}=\left(\Omega_{PA},\mathcal{F},\mu_{E},\tau_d\right)$ consists of a dynamical law $\tau_d$ on the probability space $\Gamma$.
\end{definition}

\section{The Gouy-Stodola Theorem}
Irreversibility occurs in all natural processes. In accordance with the second law of thermodynamics, irreversibility is the phenomenon which prevents from extracting the most possible work from various sources. Consequently, it prevents from doing the complete conversion of heat or energy in work; indeed, in all the natural processes a part of work, $W_\lambda$ is lost due to irreversibility. This work can be related to the entropy generation. In this Section the entropy generation and its relation to the work lost due to irreversibility is developed for the open systems, introducing the Gouy-Stodola Theorem.

\begin{definition} \cite{bejan2006}
The work lost $W_\lambda$ for irreversibility is defined as:
\begin{equation}
W_\lambda = \int_0^\tau dt\, \dot{W} _\lambda
\end{equation}
where $\dot{W} _\lambda$ is the power lost by irreversibility, defined as:
\begin{equation}
\dot{W}_\lambda = \dot{W}_{max} - \dot{W}
\end{equation}
with $\dot{W}_{max}$ maximum work transfer rate (maximum power transferred), which exists only
in the ideal limit of reversible operation, and $\dot{W}$ the effective work transfer rate (effective power transferred).
\end{definition}

\begin{definition} \cite{bejan2006}
The entropy of the whole system, composed by the open system and the environment is defined as:
\begin{equation}
\label{entropia}
S=\int\bigg(\frac{\delta Q}{T}\bigg)_{rev}=\Delta S_{e}+ S_{g}
\end{equation}
where $S_{g}$ is the entropy generation, defined as:
\begin{equation}
S_g =\int_0^\tau dt\,\dot{S}_g
\end{equation}
with $\dot{S}_g$ entropy generation rate defined as:
\begin{equation} \label{eqentrrate}
\dot{S}_g = \frac{\partial S}{\partial t} + \sum_{out}G_{out}s_{out} - \sum_{in}G_{in}s_{in} - \sum_{i=1}^N \frac{\dot{Q}_i}{T_i}
\end{equation}
while $\Delta S_{e}$ is defined as the entropy variation that would be obtained exchanging reversibly the same heat and mass fluxes throughout the system boundaries, $G$ in the mass flow, the terms $out$ and $in$ means the summation over all the inlet and outlet port, $s$ is the specific entropy, $S$ is the entropy, $\dot{Q}_i, i \in [1,N]$ is the heat power exchanged with the $i- $th heat bath and $T_i$ its temperature, $\tau$ is the lifetime of the process which occurs in the open system.
\end{definition}

 Then the term due to irreversibility, the entropy generation $S_{g}$, measures how far the system is from the state that will be attained in a reversible way.

\newtheorem{theorem}{Theorem}
\begin{theorem}[Gouy-Stodola Theorem]
In any open system, the work lost for irreversibility $W_\lambda$ and the entropy generation $S_g$ are related each another as:
\begin{equation}
W_\lambda = T_a \, S_g
\end{equation}
where $T_a$ is the ambient temperature. 
\end{theorem}
\begin{proof}
Considering the First and Second Law of Thermodynamics for the open systems, the maximum power transferred is:
\begin{equation}
\dot{W}_{max} =\sum_{in}G_{in}\bigg(h+\frac{v^2}{2}+g\,z+T_a\,s\bigg)_{in} - \sum_{out}G_{out}\bigg(h+\frac{v^2}{2}+g\,z+T_a\,s\bigg)_{out} - \frac{d}{dt}(E-T_a\,\dot{S})
\end{equation}
while the effective power transferred results:
\begin{equation}
\dot{W}=\sum_{in}G_{in}\bigg(h+\frac{v^2}{2}+g\,z+T_a\,s\bigg)_{in} - \sum_{out}G_{out}\bigg(h+\frac{v^2}{2}+g\,z+T_a\,s\bigg)_{out} - \frac{d}{dt}(E-T_a\,\dot{S})-T_a\,\dot{S}_g
\end{equation}
where $h$ is the specific enthalpy, $v$ the velocity, $g$ the gravity constant, $z$ the height and $E$ is the instantaneous system energy integrated over the control volume.

Considering the definition of power lost, $\dot{W}_\lambda$, it follows that:
\begin{equation}
\dot{W}_\lambda = T_a\,\dot{S}_g
\end{equation}
form which, integrating over the range of lifetime of the process, the Gouy-Stodola theorem is proven:
\begin{equation}
W_\lambda =\int_0^\tau dt\, \dot{W}_\lambda =T_a\,\int_0^\tau dt\, \dot{S}_g = T_a\,S_g
\end{equation}
\end{proof}

The Gouy-Stodola result on the entropy generation is expressed in a global way, without any statistical approach and expression. In order to extend the use the Gouy-Stodola theorem to any approach and context, a statistical expression of entropy generation is required. To do so the following definition can be introduced:

\begin{definition} \cite{gallavotti2006}
The entropy production $\Sigma_{prod}$ is defined as:
\begin{equation}	
\Sigma_{prod}=\sum_{i=1}^N \frac{\dot{Q}}{k_B T_i}=\int_\Gamma \Sigma (\sigma) \mu_E (d\sigma)
	\label{entropy production}
\end{equation}
with $N$ number of heat baths, whose temperature is $T_i$, $i \in [1,N]$ in contact with the system, $\dot{Q}_i$, $i \in [1,N]$ heat power exchanged with each $i-$th heat bath, $k_B$ Boltzmann constant, $\Sigma (\sigma)$ phase space contraction and $\mu_E$ SRB-statistics.
\end{definition}

\begin{theorem}
In a stationary state, the entropy generation and the entropy production are related one another by the relation:
\begin{equation}
S_g=-k_B\,\int_0^\tau dt\, \Sigma_{prod}
\end{equation}
\end{theorem}
\begin{proof}
If the system is in a stationary state ($\partial S/\partial t =0$) and if:
\begin{enumerate}
\item the system is closed: $\sum_{out}G_{out}s_{out} =0$ and $\sum_{in}G_{in}s_{in}$, or
\item the system is open, but $\sum_{out}G_{out}s_{out} - \sum_{in}G_{in}s_{in}=0$
\end{enumerate}
and considering the relations (\ref{eqentrrate}) and (\ref{entropy production}), then it follows:
\begin{equation}
\dot{S}_g=-k_B\, \Sigma_{prod}
\end{equation}
from which, integrating the statement is obtained.
\end{proof} 

\begin{theorem}
The thermodynamic Lagrangian is:
\begin{equation}
\mathcal{L}=W_\lambda
\end{equation}
\end{theorem}
\begin{proof}
Considering the entropy density per unit time and temperature $\rho_S$, the Lagrangian density per unit time and temperature $\rho_{\mathcal{L}}$, the power density per unit temperature $\rho_{\pi}$, and the dissipation function $\phi$, the following relation has been proven \cite{lavenda}:
\begin{equation}
\rho_\mathcal{L} = \rho_S -\rho_\pi -\phi
\end{equation}
and considering that $ \rho_S -\rho_\pi =2 \phi $, it follows \cite{lucia1995} that:
\begin{equation}
\rho_\mathcal{L} = \phi
\end{equation}
consequently \cite{lucia2008},
\begin{equation}
\mathcal{L} =\int_t dt\int_T dT\int V \rho_\mathcal{L}\,dV = \int_t dt\int_T dT\int V \rho_\mathcal{L}\, \phi\,dV = W_{\lambda}
\end{equation}
\end{proof}

\begin{theorem}
At the stationary state, the work $\lambda$ lost for irreversibility is an extremum.
\end{theorem}
\begin{proof}
Considering the definition of action $\mathcal{A}$ it follows that:
\begin{equation}
\mathcal{A}=\int_0^\tau dt\,\mathcal{L}=\int_0^\tau dt\,W_\lambda
\end{equation}
and considering the least action principle it follows:
\begin{equation}
\delta \mathcal{A} \leq 0 \Rightarrow \delta W_\lambda \leq 0
\end{equation}
which allow to state that:
\begin{enumerate}
\item $W_\lambda$ is minimum if the work lost is evaluated inside the system
\item $W_\lambda$ is maximum if the work lost is evaluated outside the system, inside the environment
\end{enumerate}
in accordance with the thermodynamic sign convention.
\end{proof}

Of course, this extremum is extended also the entropy generation, $S_g$, using the Gouy-Stodola Theorem $W_\lambda = T_a\,S_g$.

\section{Conclusions}
A general principle to approach the stability of the stationary states of the open systems is required in science and engineering because it would represent a new approach to the analysis of these systems, with the result of improving their applications in mathematical and theoretical biology, biotechnologies, nanotechnologies, ecology, climate changes, energy optimization, thermo-economy, phase separation, aging process, system theory and control, pattern formation, cancer pharmacology, DNA medicine, metabolic engineering, chaotic and dynamical systems, etc..

Here, the Gouy-Stodola theorem has been proved to be the searched variational principle, which satisfies the two fundamental request:
\begin{enumerate}
\item to be a work principle, because a useful variational principle out of the use of work has been proved not to be obtained \cite{lebowitz} 
\item to use only one temperature which remains constant (because the environmental temperature is always considered constant in the usual applications), because for the open systems, the entropy production for the system and the reservoir has been proved to be obtained if and only if only one heat bath \cite{lebowitz} exists and its temperature is constant \cite{lebowitz}.
\end{enumerate}

\bibliographystyle{my-h-elsevier}

\begin{thebibliography}{10}

% bibliography entries should follow the format of the sample items
% below, without the "%" at the beginning of the line.

%\bibitem{vgeemen}
%A. Albano and S. Katz,
%Van Geemen's Families of Lines on Special
%Quintic Threefolds,
%Manuscripta Math. {\bf 70} (1991) 183.

%\bibitem{AK}
%A. Albano and S. Katz,
%Lines on The Fermat Quintic Threefold,
%Trans. AMS. {\bf 324} (1991) 353.

%\bibitem{ADD}
%S.K. Ashok, E. Dell'Aquila and D.-E. Diaconescu,
%Fractional Branes in Landau-Ginzburg Orbifolds,
%hep-th/0401135.
\bibitem{ozisik1980}  M.N \"{O}zisik, Heat Conduction, John Wiley \& Sons, New York, 1980.
\bibitem{truesdell} C. Truesdell,  Rational Thermodynamics, Springer-Verlag, Berlin, 1984.
\bibitem{arnold} V.I. Arnold, Mathematical Methods of Classical Mechanics, Springer-Verlag, Berlin, 1989.
\bibitem{lucia1995} U. Lucia, Mathematical consequences and Gyarmati's principle in Rational Thermodynamics, Il Nuovo Cimento B, \textbf{110}, 10 (1995) 1227--1235.
\bibitem{martyushev} L.M. Martyushev and V.D. Seleznev, Maximum entropy production principle in physics, chemistry and biology, Phys. Rep., \textbf{426} (2006) 1--45.
\bibitem{lebowitz} J.L. Lebowitz, Boltzmann's Entropy and Large Deviation Lyapunov Functionals for Closed and Open Macroscopic Systems, Los Alamos 1112.1667, 2011: http://www.math.rutgers.edu/$\sim$lebowitz/PUBLIST/jllpub-559.pdf
\bibitem{gouya} G. Gouy, Sur les transformation et l'équilibre en Thermodynamique, Comptes Rendus de l'Acadèmie des Sciences Paris, {\bf 108} (10) (1889) 507--509.
\bibitem{gouyb}  G. Gouy,  Sur l'énergie utilisable, Journal de Physique, {\bf 8} (1889) 501--518.
\bibitem{duhem} P. Duhem, Sur les transformations et l'\'{e}quilibre en Thermodynamique. Note de M.P. Duhem, Comptes Rendus de l'Acad\`{e}mie des Sciences Paris, {\bf 108} (13) (1889) 666--667.
\bibitem{gouyc} G. Gouy, Sur l'\'{e}nergie utilisable et le potentiel thermodynamique. Note de M. Gouy, Comptes Rendus de l'Acad\`{e}mie des Sciences Paris, {\bf 108} (10) (1889): 794.
\bibitem{lucia2008} U. Lucia, Probability, ergodicity, irreversibility and dynamical systems,. Proc. R. Soc. A \textbf{464}, 2093 (2008) 1089--1184.
\bibitem{huang1987} K. Huang, Statistical Mechanics, John Wiley \& Sons, New York, 1987.
\bibitem{gallavotti2003} G. Gallavotti, SRB distribution for Anosov maps. Two lectures on the construction of the SRB distributions for Anosov maps, Cargese summer school 18-31 august 2003: http://ipazia.rutgers.edu/$\sim$ giovanni/gpub.html\#E.
\bibitem{lucia2010} U. Lucia, Maximum entropy generation and $\kappa-$exponential model, Physica A {\bf 389} (2011) 4558--4563.
\bibitem{garcia2006} V. Garc\'{\i}a-Morales and J. Pellicer, Microcanonical foundation of nonextensivity and generalized thermostatistics based on the fractality of the phase space, \textit{Physica A}, \textbf{361} (2006)  161--172.
\bibitem{hoover1998} W.G. Hoover, Liouville's theorems, Gibbs' entropy, and multifractal distributions for nonequilibrium steady states, J. Chem. Phys., \textbf{109} (1998) 11, 4164--4170.
\bibitem{billingsley1979} P. Billingsley, Probability and Measure, John Wiley \& Sons, New York, 1979.
\bibitem{bejan2006} A. Bejan, Advanced Engineering Thermodynamics,. John Wiley \& Sons, New York, 2006.
\bibitem{gallavotti2006} G. Gallavotti, Entropy, thermostats and chaotic hypothesis,. Chaos {\bf 16} (2006)  384--389.
\bibitem{lavenda} B.H. Lavenda, Thermodynamics of Irreversible Processes. Macmillan Press, London, 1978.

\end{thebibliography}

\end{document}